%% file: Tempo.tex
\documentclass{easychair}   
\usepackage[T1]{fontenc}
\usepackage[utf8]{inputenc}
\usepackage{amsmath,amssymb}
\usepackage{mathtools}
\usepackage{stmaryrd}   
\usepackage{siunitx}
\usepackage{paperstyle}
\usepackage{tikz}
\usetikzlibrary{arrows.meta}
\usetikzlibrary{arrows.meta,decorations.pathreplacing}
\usetikzlibrary{positioning}
\usetikzlibrary{fit,backgrounds}          
\colorlet{cNegBg}{cyan!8}                 
\colorlet{cRtBg}{violet!8}                
\newcommand{\reff}[1]{{\scriptsize\textcolor{black!50}{#1}}}

\usepackage{multirow}
\usepackage{enumitem}
\usepackage{xcolor}         
\usepackage{tabularx}
\usepackage{doi}
\usepackage{amsthm}
\usepackage[capitalise,noabbrev]{cleveref}   
\theoremstyle{plain}
\usepackage{aliascnt}
\newtheorem{theorem}{Theorem}[section]
\newaliascnt{example}{theorem}
\newtheorem{example}[example]{Example}
\aliascntresetthe{example}
\crefname{example}{Example}{Examples}
\newaliascnt{lemma}{theorem}
\newtheorem{lemma}[lemma]{Lemma}
\aliascntresetthe{lemma}
\newaliascnt{proposition}{theorem}
\newtheorem{proposition}[proposition]{Proposition}
\aliascntresetthe{proposition}
\newaliascnt{corollary}{theorem}
\newtheorem{corollary}[corollary]{Corollary}
\aliascntresetthe{corollary}
\theoremstyle{definition}
\newaliascnt{definition}{theorem}
\newtheorem{definition}[definition]{Definition}
\aliascntresetthe{definition}
\newaliascnt{remark}{theorem}

\aliascntresetthe{remark}
\newaliascnt{assumption}{theorem}

\aliascntresetthe{assumption}

\newcommand{\Inst}{\mathsf{Inst}}

\newcommand{\Dur}{\mathsf{Dur}}
\newcommand{\timeT}{\mathbb{T}}
\newcommand{\durD}{\mathbb{D}}

\newcommand{\sortof}{\operatorname{sort}}
\newcommand{\Per}{\mathsf{Per}}
\newcommand{\ctx}{\rho}
\newcommand{\now}{\mathit{now}}
\newcommand{\trig}{\mathit{trig}}
\newcommand{\dstart}{\mathit{start}}

\newcommand{\diff}{\delta}

\newcommand{\tier}{\operatorname{tier}}
\newcommand{\Ord}{\mathsf{Ord}}
\newcommand{\Diff}{\mathsf{Diff}}
\newcommand{\Mod}{\mathsf{Mod}}
\newcommand{\Trace}{\mathsf{Tr}}

\newcommand{\Policy}{\mathcal{P}}
\newcommand{\PolicyOffer}{\Policy_{\text{offer}}}
\newcommand{\PolicyReq}{\Policy_{\text{req}}}

\newcommand{\sem}[1]{\llbracket #1 \rrbracket}
\newcommand{\refines}{\preceq} 
\newcommand{\Verdicts}{\mathbb{V}}
\newcommand{\verdict}{\operatorname{verdict}}
\newcommand{\prodV}{V_{\mathrm{prod}}}

\newcommand{\odrl}[1]{\textsf{#1}}

\crefname{definition}{Definition}{Definitions}
\Crefname{definition}{Definition}{Definitions}
\crefname{theorem}{Theorem}{Theorems}
\Crefname{theorem}{Theorem}{Theorems}
\crefname{lemma}{Lemma}{Lemmas}
\Crefname{lemma}{Lemma}{Lemmas}
\crefname{proposition}{Proposition}{Propositions}
\Crefname{proposition}{Proposition}{Propositions}
\crefname{corollary}{Corollary}{Corollaries}
\Crefname{corollary}{Corollary}{Corollaries}
\crefname{remark}{Remark}{Remarks}
\Crefname{remark}{Remark}{Remarks}
\crefname{example}{Example}{Examples}
\Crefname{example}{Example}{Examples}
\crefname{assumption}{Assumption}{Assumptions}
\Crefname{assumption}{Assumption}{Assumptions}

\begin{document}
\title{Sort-Stratified Semantics for Temporal Conflict Detection in ODRL Policies}
\titlerunning{Sort-Stratified ODRL Temporal Semantics}
\author{
Daham M. Mustafa\inst{1}\inst{2}
\and Diego Collarana\inst{2}
\and Sabrina Kirrane\inst{3}
\and Christoph Lange\inst{1}\inst{2}
\and Christoph Quix\inst{1}\inst{2}
\and
Sandra Geisler
\inst{1}
\and Stefan Decker\inst{1}\inst{2}
\and Rafiqul Haque\inst{4}
}
\authorrunning{D. M. Mustafa et al.}
\institute{
  RWTH Aachen University, Aachen, Germany\\
\and
  Fraunhofer FIT, Sankt Augustin, Germany
\and
  Vienna University of Economics and Business (WU), Vienna, Austria
\and
University of Galway, Galway, Ireland and GATE Institute, Sofia, Bulgaria
}
\maketitle
\begin{abstract}
In the Open Digital Rights Language (ODRL), temporal constraints range over two sorts, instants and durations, but the comparison operators do not distinguish them. The same operator thus means ``earlier instant'' or ``shorter duration,'' leaving conflict detection between two policies unsound. We resolve this by sort stratification: each temporal operand is typed to one of two ordered domains, points in time or amounts of time. Each constraint then denotes an interval, and conflict reduces to interval comparison under a three-valued verdict (Conflict, Compatible, Unknown). We characterise the check's decidability across a static and a runtime fragment, prove it sound,
and evaluate it on a benchmark of policy problems compiled to TPTP and SMT-LIB, available as an artefact.\footnote{\url{https://github.com/Daham-Mustaf/odrl-temporal-benchmark}}
\end{abstract}
\smallskip
\noindent\textbf{Keywords:} Data spaces,
Policy Conflict Detection, Temporal Constraints, Interval Semantics,
Denotational Semantics, Three-Valued Semantics
\section{Introduction}
\label{sec:introduction}
Data spaces~\cite{otto2022designing, drees2021mobility} let organisations exchange digital assets under
explicit usage policies, which commonly carry temporal conditions, e.g., a policy may expire on a fixed
date, impose an embargo before first use, limit the total usage time, or grant access only on a
recurring schedule. These policies are written in the Open Digital Rights Language (ODRL)~2.2~\cite{odrl-model}, and adopted across Data spaces specifications~\cite{dssc2024blueprint}. For example, in Eclipse Dataspace Connector (EDC)~\cite{EDC2026} a provider publishes an \emph{offer} and a
consumer submits a \emph{request}; during contract negotiation~\cite{dsp2025} the two are checked
for \emph{compatibility}: an \emph{agreement} forms only when the offer and the request can be met
together. 

This check precedes any access;
enforcement then evaluates the agreement at runtime. ODRL expresses temporal conditions as
\emph{constraints} over dedicated operands: the instant \odrl{dateTime}; the duration-based
\odrl{delayPeriod}, \odrl{elapsedTime}, and \odrl{meteredTime}; and the recurrence
\odrl{timeInterval}~\cite{ODRLVocab}. Logical operators combine these constraints. To our knowledge, no ODRL formalisation or engine gives the duration or recurrence operands semantics; and although Salas et al.~\cite{salas2025evaluationcomparisonsemanticsodrl,salas2026normalisation}
and Bonatti et al.~\cite{bonatti2025towards} already compare two policies, they do so over
\odrl{dateTime} alone, as a sortless attribute. We give these operands semantics and check
compatibility between an offer and a request over the full temporal operand set.

Comparison operators do not distinguish between  instants and durations sorts, so the same \odrl{lt} operator has different interpretations: on an instant it means “earlier than”, as in $(\odrl{dateTime}, \odrl{lt}, \textsf{2026-12-31})$; on a duration it means “shorter than”, as in $(\odrl{elapsedTime}, \odrl{lt}, \textsf{P30D})$, where \textsf{P30D} denotes thirty days in ISO 8601. Because of this ambiguity, the operator alone does not determine the intended interpretation, nor does it prevent comparisons between incomparable sorts. Introducing a sort for each operand resolves this ambiguity by fixing the interpretation of operators and excluding comparisons between instants and durations. 

Additionally, the operands are not independent. Metered usage, the time actually spent using an asset, can never exceed the elapsed time since access began~\cite{ODRLVocab}; and under \odrl{andSequence}, which requires constraints to be satisfied in order, a \odrl{delayPeriod} is counted from the moment the preceding constraint is first satisfied, not from a fixed date. Thus
$\odrl{andSequence}(\odrl{dateTime}~\odrl{gteq}~\textsf{2026-01-01},\ \odrl{delayPeriod}~\odrl{gteq}~\textsf{P1D})$
permits the action only at least one day after the \textsf{2026-01-01} constraint is first satisfied. Checking each operand against its own constraint in isolation, therefore, misses conflicts that arise from these relations: an offer's $\odrl{meteredTime}~\odrl{eq}~\textsf{P30D}$ and a request's $\odrl{elapsedTime}~\odrl{lteq}~\textsf{P10D}$ are each satisfiable alone but jointly unsatisfiable, since thirty days of metered usage necessarily require at least thirty elapsed days, which the request's limit of ten forbids.

In this paper, we propose a sort-stratified semantics for ODRL's temporal constraints, and a conflict
detector built on it. The semantics gives each operand a precise
meaning, separating time points from durations, so that constraints the current operators read ambiguously gain a single interpretation. The detector decides whether an offer and a request are
compatible over their temporal terms, returning one of three answers, “compatible”, “in conflict”, or
“undetermined” when only one side constrains a condition; it catches conflicts spanning two
operands, such as metered use never exceeding elapsed time, handles the recurrence and ordered
\odrl{andSequence} operators, and classifies each conflict by the
reasoning it needs. The same semantics covers conflict detection at both negotiation time and runtime, so a conflict
found before agreement provably holds at runtime. The benchmark encodes
conflict-detection problems in TPTP and SMT-LIB and decides them with automated reasoners. The remainder of the paper covers preliminaries (\cref{sec:prelim}), the sort-stratified
denotational semantics (\cref{sec:problem}), the conflict-detection framework (\cref{sec:conflict}), the evaluation (\cref{sec:eval}),
 related work (\cref{sec:related}) and the conclusion
(\cref{sec:conclusion}).
\vspace{-10pt}
\section{Preliminaries and Motivating Example}
\label{sec:prelim}

An ODRL policy~\cite{odrl-model} comprises \emph{rules}, each of which is a permission,
prohibition, or obligation concerning an action a party performs on an asset. A rule may carry zero or more constraints and applies only when all of them are satisfied~\cite{ODRLVocab}. We
consider only the \emph{temporal constraints}, those whose left operand is a temporal operand and
which constrain the timing and duration of asset usage; non-temporal constraints and the
additional rule constructs (duties, remedies, and consequences) lie \emph{outside our scope}.

\noindent A \emph{constraint} is a triple $(\ell,\bowtie,v)$ of a \emph{left operand}
$\ell$, an \emph{operator} $\bowtie$, and a \emph{right operand}
$v$~\cite{odrl-model}. ODRL provides twelve operators for specifying atomic constraints, six comparison operators,
\odrl{eq}~($=$), \odrl{neq}~($\neq$), \odrl{lt}~($<$), \odrl{lteq}~($\leq$),
\odrl{gt}~($>$), \odrl{gteq}~($\geq$), and six set and membership operators
($\odrl{isA},\odrl{hasPart},\odrl{isPartOf},\odrl{isAllOf},\odrl{isAnyOf},\odrl{isNoneOf}$), ODRL combines constraints into \emph{logical constraints} through four operators: \odrl{and},
satisfied when all of them hold; \odrl{or}, when at least one holds; \odrl{xone}, when exactly
one holds; and \odrl{andSequence}, when all hold in the listed order~\cite{ODRLVocab}. The specification singles out \odrl{andSequence} as a source of temporal requirements between
its operands and notes that these can leave a policy unresolvable~\cite{ODRLVocab}.
\input{fig-bsb-denotations}
\noindent \odrl{dateTime} ranges over \emph{instants}, i.e., points on a timeline (ISO 8601
dates normalised to UTC, e.g.\ \textsf{2026-12-31}); \odrl{delayPeriod}, \odrl{elapsedTime},
\odrl{meteredTime}, and \odrl{timeInterval} range over \emph{durations}, i.e., amounts of time (ISO 8601 spans, e.g.\ \textsf{P30D}, \textsf{PT12H}). Three further left operands are temporal but do not lie on this global timeline, so they are out of scope~\cite{ODRLVocab}. \odrl{event} 
whose temporal extent is determined by an external event, not by a date or span~\cite{mustafa2026denotationalsemanticsodrlknowledgebased}. \odrl{absoluteTemporalPosition} and
\odrl{relativeTemporalPosition}
which refer to positions within a media asset's own timeline rather than calendar time (e.g.\ seconds 192 to 250) or a fraction of the stream (e.g.\ 33\% to 48\%)~\cite{ODRLVocab}.

\bigskip
\noindent\textbf{Motivating Example}
\label{sec:motivation}
In the Culture Dataspace, the provider, the Bavarian State Library (BSB), offers
a digitised manuscript (\textsf{drk:bsb-clm-14000}), valid from 1 January 2026,
under five conditions: it may be accessed only before 31 December 2026, used within a thirty-day total window, with exactly thirty days of metered (actual)
use, first used only one day after access is granted, and only on dates recurring
every thirty days from 1 January 2026. The consumer, the French National Library
(BnF), requests the same manuscript under its own conditions: access on 1 June
2027, only on dates recurring every forty-five days from that date, at most ten
days of use in total, and immediate first use, with no delay.

\begin{example}[BSB--BnF temporal access]\label{ex:bsb}
The BSB offer $\PolicyOffer$ has five constraints:
$\odrl{dateTime}~\odrl{lt}~\textsf{2026-12-31}$,
$\odrl{elapsedTime}~\odrl{lteq}~\textsf{P30D}$,
$\odrl{meteredTime}~\odrl{eq}~\textsf{P30D}$,
$\odrl{delayPeriod}~\odrl{gteq}~\textsf{P1D}$, and
$\odrl{timeInterval}~\odrl{eq}~\textsf{P30D}$ (anchored at \textsf{2026-01-01}).
The BnF request $\PolicyReq$ has four constraints:
$\odrl{dateTime}~\odrl{eq}~\textsf{2027-06-01}$,
$\odrl{elapsedTime}~\odrl{lteq}~\textsf{P10D}$,
$\odrl{delayPeriod}~\odrl{eq}~\textsf{PT0S}$, and
$\odrl{timeInterval}~\odrl{eq}~\textsf{P45D}$ (anchored at \textsf{2027-06-01}),
with no constraint on \odrl{meteredTime}.
\end{example}
\noindent
\cref{tab:example-constraints} shows the offer and request constraints for each temporal
operand as intervals, together with the per-operand verdict. The policies \Conflict\ on
\odrl{dateTime}, \odrl{delayPeriod}, and \odrl{timeInterval}, are \Compatible\ on
\odrl{elapsedTime}, and leave \odrl{meteredTime} \Unknown, since only the offer constrains it. One conflict, however, is invisible operand by operand: the offer's
$\odrl{meteredTime}~\odrl{eq}~\textsf{P30D}$ and the request's
$\odrl{elapsedTime}~\odrl{lteq}~\textsf{P10D}$ cannot hold together, since metered use cannot exceed the ten elapsed days the request allows. Identifying both the per-operand conflicts and
this cross-operand one before an agreement is formed is the task of the policy conflict
detector we present. Each verdict then governs the next step in negotiation: \Compatible\
permits an agreement, \Conflict\ requires at least one party to revise its terms, and
\Unknown\ leaves the result undetermined and calls for a counter-offer or \mbox{further information}.
\vspace{-10pt}
\section{Temporal Sorts and Interval Semantics}
\label{sec:problem}
\input{pipe-line}
\Cref{fig:pipeline} summarises the framework: a policy pair is typed to instant or
duration sorts, each constraint is denoted as an interval (a periodic set for
\odrl{timeInterval}), the denotations form a product tested against the background
theory $\Phi$, a tier fixes the decision procedure, and the per-operand verdicts
aggregate by $\min$ into \Conflict, \Compatible, or \Unknown, which drives the
negotiation step; on a \Compatible\ agreement the framework moves to runtime
enforcement (the dashed box), and the two background tints separate the static and
runtime phases. The
two sorts differ algebraically: durations add and have a zero, instants are only ordered, and they meet through a shift, where a duration moves an instant, and a
difference, where two instants give a duration. 
\begin{definition}[Temporal structures]\label{def:structures}
The \emph{timeline} $\timeT=(T,\leq_{\timeT})$ is a linear order that is
\emph{affine} over $(\mathbb{Q},+)$: fixing any origin identifies
$(T,\leq_{\timeT})$ order-preservingly with $(\mathbb{Q},\leq)$. The \emph{duration domain} $\durD=(D,+,0,\leq_{\durD})$ is a totally ordered
cancellative commutative monoid embedding order-preservingly into
$(\mathbb{Q}_{\geq0},+,\leq)$ and closed under the non-negative differences of instants. For an instant $s\in T$ and $d\in\mathbb{Q}$, the
\emph{shift} $s+d\in T$ is the instant $d$ after $s$; the \emph{difference}
$t-s\in\mathbb{Q}$ is its inverse, and the \emph{elapsed duration}
$\diff(s,t)=t-s$ lies in $D$ exactly when $s\leq_{\timeT}t$.
\end{definition}
\noindent
We exclude \textsf{xsd:yearMonthDuration} because month- and year-based durations do not denote fixed-length intervals. Consequently the full \textsf{xsd:duration} datatype is only partially ordered, since values such as \textsf{P1M} and \textsf{P30D} are not comparable. We build the semantics on intervals, so we first consider only the operators whose constraints denote intervals. 
A comparison such as \emph{before a value}, \emph{at least a value}, or
\emph{exactly a value} selects an order-convex subset of the operand's domain
(i.e., whenever it contains two points, it also contains all intermediate points).
\begin{definition}[Interval operators]\label{def:int-op}
On a totally ordered set $(D,\leq)$, an operator $\bowtie$ is an
\emph{interval operator} if for every $v\in D$ the solution set
$\{x\in D\mid x\bowtie v\}$ is order-convex. Among the ODRL operators the
interval operators are exactly
$\mathcal{O}_I=\{\odrl{lt},\odrl{lteq},\odrl{gt},\odrl{gteq},\odrl{eq}\}$.
\end{definition}
\noindent
Each interval operator is interpreted on the ordered domain fixed by its operand's sort; hence, every temporal operand must be assigned a sort to fix the meaning of its constraints.
\begin{definition}[Temporal sorts]\label{def:sort}
A \emph{temporal sort} is $\sigma\in\{\Inst,\Dur\}$: the instant sort $\Inst$,
interpreted over $\timeT$, and the duration sort $\Dur$, interpreted over
$\durD$. The sort fixes the totally ordered domain $(D_\ell,\leq)$ that an
operand $\ell$'s value ranges over and on which the interval operators
$\mathcal{O}_I$ are interpreted. Each self-contained left operand has a unique
$\sortof(\ell)$, with $\sortof(\odrl{dateTime})=\Inst$ and
$\sortof(\ell)=\Dur$ for
$\ell\in\{\odrl{delayPeriod},\odrl{elapsedTime},\odrl{meteredTime},
\odrl{timeInterval}\}$.
\end{definition}
\noindent
The ODRL vocabulary recommends a restricted set of operators for each temporal operand~\cite{ODRLVocab}: all interval operators for \odrl{dateTime}, but only specified subsets for the duration and recurrence operands. It recommends neither \odrl{neq} nor any set-based operator for temporal operands.
\begin{definition}[Admissible constraint]\label{def:admissible}
Let $\mathit{Adm}(\ell)$ be the operator set for each operand $\ell$:
\begin{align*}
\mathit{Adm}(\odrl{dateTime})&=\mathcal{O}_I,\;
\mathit{Adm}(\odrl{delayPeriod})=\{\odrl{eq},\odrl{gt},\odrl{gteq}\},\; 
\mathit{Adm}(\odrl{elapsedTime})=\mathit{Adm}(\odrl{meteredTime})
=\{\odrl{eq},\odrl{lt},\odrl{lteq}\},\\
\mathit{Adm}(\odrl{timeInterval})&=\{\odrl{eq}\}.
\end{align*}
A constraint $(\ell,\bowtie,v)$ is \emph{admissible} if
$\operatorname{\bowtie}\in\mathit{Adm}(\ell)$, $v\in D_\ell$, and
$\{x\in D_\ell\mid x\bowtie v\}\neq\emptyset$.
\end{definition}
\noindent
By \cref{def:int-op}, every admissible constraint denotes a nonempty order-convex subset of $D_\ell$.  For all operands except \odrl{timeInterval}, this set is an interval of the operand's
domain; we call these four operands \emph{interval-valued} and \odrl{timeInterval} the
\emph{periodic} operand. Policy semantics, satisfiability, and conflict detection are
defined in terms of these denotations.
\begin{definition}[Sort-indexed interval denotation]\label{def:denotation}
For an admissible $c=(\ell,\bowtie,v)$ with $\ell$ interval-valued and
$\sortof(\ell)=\sigma$, write $\leq$ for $\leq_{\timeT}$ when $\sigma=\Inst$ and
for $\leq_{\durD}$ when $\sigma=\Dur$. The denotation is the solution set
$\sem{\ell\,\bowtie\,v}=\{x\in D_\ell\mid x\bowtie v\}$:
\[
\begin{array}{@{}cccccc@{}}
\toprule
\bowtie & \odrl{eq} & \odrl{lteq} & \odrl{gteq} & \odrl{lt} & \odrl{gt}\\
\midrule
\sem{\ell\,\bowtie\,v} & \{v\} & (\bot,v] & [v,\top) & (\bot,v) & (v,\top)\\
\bottomrule
\end{array}
\]
A square bracket includes its endpoint and a round bracket excludes it; $\bot$ and
$\top$ are not elements of $D_\ell$ but unbounded markers ordered below and above
every element, so a bracket at $\bot$ or $\top$ is round. On $\durD$ the least
element is $0$, so $(\bot,v]$ and $(\bot,v)$ are $[0,v]$ and $[0,v)$.
\end{definition}
\noindent
In \cref{ex:bsb}, we obtain the following denotations:  
$\odrl{dateTime}\,\odrl{lt}\,\textsf{2026-12-31}$ denotes $(\bot,\textsf{2026-12-31})$,  
$\odrl{elapsedTime}\,\odrl{lteq}\,\textsf{P30D}$ denotes $[0,\textsf{P30D}]$,  $\odrl{delayPeriod}\,\odrl{gteq}\,\textsf{P1D}$ denotes $[\textsf{P1D},\top)$, and  
$\odrl{dateTime}\,\odrl{eq}\,\textsf{2027-06-01}$ denotes $\{\textsf{2027-06-01}\}$.
\vspace{-10pt}\subsection{Periodic Constraints}
\label{sec:periodic-constraints}
ODRL's \odrl{timeInterval} describes a condition that recurs after a fixed
period~\cite{ODRLVocab}. Its period is a duration, so \odrl{timeInterval} is
duration-sorted, yet its denotation is a set of instants rather than a duration
interval: the reference instant shifted repeatedly by the period. In \cref{ex:bsb}
the offer reopens every \textsf{P30D} from \textsf{2026-01-01}.
\begin{definition}[Periodic \odrl{timeInterval}]\label{def:ti-rec}
For an anchor $a\in\timeT$ and a period $p\in\durD$ with $p>0$, the
\emph{periodic set} is
\[
\Per(a,p)=\{\,a+kp\mid k\in\mathbb{Z}\,\}\subseteq\timeT .
\]
For a constraint $(\odrl{timeInterval},\odrl{eq},p)$ in a policy $P$ with anchor
$\dstart_P$, its denotation is
$\sem{(\odrl{timeInterval},\odrl{eq},p)}^{P}_{\ctx}=\Per(\dstart_P,p)$, and a
context $\ctx$ satisfies it when $\ctx.\now\in\Per(\dstart_P,p)$.
\end{definition}

\noindent
In the case of periodic temporal expressions, only \odrl{eq} is admissible
(\cref{def:admissible}). The integer $k$ indexes the occurrences: in \cref{ex:bsb}
the offer's set $\Per(\textsf{2026-01-01},\textsf{P30D})$ places $k=0$ at
\textsf{2026-01-01} and $k=1$ at \textsf{2026-01-31}, with larger $k$ later and
negative $k$ earlier. Letting $k$ range over $\mathbb{Z}$ makes each periodic set an
arithmetic progression with no first or last occurrence, so two sets are disjoint
not by which starts earlier but by whether the offset between their anchors is a
multiple of the step the periods share, a single divisibility condition. A
forward-only reading would instead add the guard $\ctx.\now\geq\dstart_P$, dropping
the negative-$k$ occurrences without changing any verdict. A \emph{context} $\ctx$ fixes the instant $\ctx.\now$, the trigger $\ctx.\trig$, the
access-begin $a_0$, and each policy's start $\dstart_P$ (all in $\timeT$); $\sem{\ell}_\ctx\in
D_\ell$ is the value of $\ell$ in $\ctx$.
\begin{lemma}[Recurrence conflict]\label{lem:rec}
For periodic constraints with anchors $a_1,a_2$ and periods $p_1,p_2$,
\[
\Per(a_1,p_1)\cap\Per(a_2,p_2)\neq\emptyset
\iff \gcd(p_1,p_2)\mid(a_2-a_1),
\]
where $\gcd(p_1,p_2)$ is the positive generator of the subgroup
$p_1\mathbb{Z}+p_2\mathbb{Z}$ of $(\mathbb{Q},+)$, and $g\mid\Delta$ means
$\Delta/g\in\mathbb{Z}$.
When the intersection is nonempty, it is itself periodic,
$\Per(a_1,p_1)\cap\Per(a_2,p_2)=\Per(a^{\ast},\operatorname{lcm}(p_1,p_2))$, where
$a^{\ast}$ is any common occurrence, obtained from the extended Euclidean algorithm
applied to $p_1,p_2$ and $a_2-a_1$.
\end{lemma}
\begin{proof}
The intersection is nonempty iff $a_2-a_1=kp_1-jp_2$ for some
$k,j\in\mathbb{Z}$, that is $a_2-a_1\in p_1\mathbb{Z}+p_2\mathbb{Z}$. This
subgroup of $(\mathbb{Q},+)$ is finitely generated, hence cyclic with positive
generator $\gcd(p_1,p_2)$, so membership is exactly
$\gcd(p_1,p_2)\mid(a_2-a_1)$.
\end{proof}
\noindent
\begin{lemma}[Windowed recurrence]\label{lem:windowed-rec}
Let $I$ be an interval of $\timeT$ with ends $\ell\in\timeT\cup\{-\infty\}$,
$u\in\timeT\cup\{+\infty\}$, and let $\Per(a,p)$ have $p>0$. Then
$I\cap\Per(a,p)\neq\emptyset$ iff some $k\in\mathbb{Z}$ has $a+kp\in I$. For bounded $I$
this is $\lceil(\ell-a)/p\rceil\le\lfloor(u-a)/p\rfloor$ (strict where the matching end of
$I$ is open); for $u=+\infty$ it always holds, dually for $\ell=-\infty$. The condition is
Presburger-definable.
\end{lemma}
\begin{proof}
$a+kp\in[\ell,u]$ iff $(\ell-a)/p\le k\le(u-a)/p$ for $p>0$; an integer in that range
exists iff $\lceil(\ell-a)/p\rceil\le\lfloor(u-a)/p\rfloor$, and if $u=+\infty$ any
$k\ge\lceil(\ell-a)/p\rceil$ works. For an open end the bound is strict: an open left end at $\ell$ gives $k\ge\lfloor(\ell-a)/p\rfloor+1$ and an open right end at $u$ gives $k\le\lceil(u-a)/p\rceil-1$, in place of $\lceil(\ell-a)/p\rceil$ and $\lfloor(u-a)/p\rfloor$. Floor and divisibility over a fixed period are
Presburger-definable.
\end{proof}
\noindent
Consider the application of \cref{lem:rec} to \cref{ex:bsb}. The offer recurs every \textsf{P30D} from
\textsf{2026-01-01} and the request every \textsf{P45D} from \textsf{2027-06-01}. The two anchors
are $516$ days apart, and the shared step is $\gcd(30,45)=15$ days. Since $516$ is not a multiple
of $15$, the periodic sets are disjoint and the recurrences \Conflict. Had the offset been a
multiple of $15$ days, the sets would share an instant and be \Compatible.
\vspace{-10pt}
\subsection{Conflict Criterion and the Product Denotation}
\label{sec:conflict-criterion}
When multiple constraints apply to one operand, its denotation is the set of values
satisfying all of them, i.e., the intersection of their intervals. Among temporal
operands, only \odrl{dateTime} is bounded from both sides; the duration operands
take a single upper or lower bound.
\begin{lemma}[Per-operand normalisation]\label{lem:normalisation}
For an interval-valued operand $\ell$ and a finite set $C$ of admissible constraints
conjoined under \odrl{and}, the \emph{normalised interval}
\[
I_\ell(C)=\bigcap_{c\in C,\ \operatorname{operand}(c)=\ell}\sem{c}
\]
(taken as $D_\ell$ when no constraint of $C$ has operand $\ell$) is either an
interval over $D_\ell$ or empty.
\end{lemma}
\begin{proof}
A finite intersection of intervals of a totally ordered set is an interval or empty.
\end{proof}
\noindent
Two intervals on one operand are disjoint when one lies entirely below the other,
which we check by comparing the upper bound of one against the lower bound of the
other.
\begin{definition}[Bound order]\label{def:precedence}
For an upper bound $u$ and a lower bound $l$, we write $u\prec l$ iff $u<l$, or
$u=l$ and at least one of the corresponding intervals excludes the boundary point.
\end{definition}
\noindent
The second case, $u=l$ covers boundary-touching intervals, which are disjoint if at least one excludes the shared boundary point.
\begin{theorem}[Conflict criterion]\label{thm:criterion}
For admissible $c_1,c_2$ on the same operand, with lower bounds $l_1,l_2$ and upper
bounds $u_1,u_2$,
\[
\sem{c_1}\cap\sem{c_2}=\emptyset\iff u_1\prec l_2\ \vee\ u_2\prec l_1.
\]
\end{theorem}
\begin{proof}
Each denotation is a nonempty interval of a totally ordered domain (\cref{def:denotation}
under admissibility), and two such intervals are disjoint iff one lies wholly below the
other, that is, an upper bound falls before the other lower bound under $\prec$.
\end{proof}
\noindent
A policy bounds several operands at once. Because distinct operands range over
independent domains, the region satisfying all constraints is the product of the
per-operand intervals.
\begin{definition}[Instant-axis denotation]\label{def:now-denotation}
\odrl{dateTime} and \odrl{timeInterval} both constrain $\ctx.\now$, so they share the
instant axis. The \emph{instant-axis denotation} of a policy $C$ is $N(C)=W(C)\cap S(C)$,
where $W(C)=I_{\odrl{dateTime}}(C)$ is the normalised \odrl{dateTime} interval
(\cref{lem:normalisation}; the full timeline if absent) and $S(C)$ is the periodic set of
$C$'s \odrl{timeInterval} (\cref{def:ti-rec}; the full timeline if absent).
\end{definition}
\noindent
The date and recurrence conditions have already been combined into one instant axis, so every
remaining condition concerns a duration, and durations vary independently of one another. A
policy's denotation is then the instant axis together with the range allowed for each duration,
satisfied only when all of these conditions hold together.
\begin{definition}[Product denotation]\label{def:product-denotation}
For a finite conjunction $C$, its denotation is the product of the shared instant
axis with the duration operands,
\[
\sem{C}=N(C)\times I_{\odrl{elapsedTime}}(C)\times I_{\odrl{meteredTime}}(C)
\times I_{\odrl{delayPeriod}}(C),
\]
where $N(C)$ is the instant-axis denotation of \cref{def:now-denotation}. A value
vector lies in $\sem{C}$ exactly when its instant component lies in $N(C)$ and each
duration component lies in its operand's normalised interval. \odrl{dateTime} and
\odrl{timeInterval} are not separate product axes: both constrain $\ctx.\now$ and
enter only through $N(C)$.
\end{definition}
\noindent
In \cref{ex:bsb} the offer's denotation is
$N(\PolicyOffer)\times[0,\textsf{P30D}]\times\{\textsf{P30D}\}\times[\textsf{P1D},\top)$ on the
instant, \odrl{elapsedTime}, \odrl{meteredTime}, and \odrl{delayPeriod} axes, where
$N(\PolicyOffer)=(\bot,\textsf{2026-12-31})\cap\Per(\textsf{2026-01-01},\textsf{P30D})$ folds the
\odrl{dateTime} window and the \odrl{timeInterval} recurrence onto one axis. Conversely, every product of realizable intervals comes from some
conjunction, so the denotation expresses exactly these product regions and no others (\cref{thm:product-expressibility}).
\begin{theorem}[Product expressibility]\label{thm:product-expressibility}
Let $I_1\times\cdots\times I_n$ have each $I_k$ a nonempty interval of $D_{\ell_k}$
realizable by admissible constraints on $\ell_k$ (\cref{def:admissible}). Then some
conjunction $C$ of admissible constraints has $\sem{C}=I_1\times\cdots\times I_n$.
\end{theorem}
\vspace{-10pt}
\subsection{Background Constraints and Tiered Decision}
\label{sec:tiers}
A policy's constraints permit certain executions, but some properties hold in every execution: a duration is never negative, and metered usage never exceeds elapsed time. We collect these properties into a background theory, so conflict detection asks not just whether two policies overlap but whether their overlap is consistent with the background theory.
\begin{definition}[Background theory and conflict]\label{def:frame}
Write $\sem{\ell}_\ctx\in D_\ell$ for the value operand $\ell$ takes in a context
$\ctx$; let $C_i$ be the constraint set of policy $P_i$ and $a_0$ the instant access
is granted, shared by the two policies. The \emph{background theory} is
\[
\begin{aligned}
\Phi={}&\{0\leq_\durD\sem{op}_\ctx \mid op\text{ a duration operand}\}
\cup\{\,\sem{\odrl{meteredTime}}_\ctx\leq_\durD\sem{\odrl{elapsedTime}}_\ctx\,\}\\
&\cup\{\,\sem{\odrl{delayPeriod}}_\ctx\leq_\durD\sem{\odrl{elapsedTime}}_\ctx\,\}
\cup\{\,a_0\leq_\timeT\ctx.\trig\leq_\timeT\ctx.\now\,\}.
\end{aligned}
\]
Constraint sets $C_1,C_2$ \emph{conflict} when $C_1\wedge C_2\wedge\Phi$ is
unsatisfiable, equivalently $\sem{C_1}\cap\sem{C_2}\cap\sem{\Phi}=\emptyset$, where
$\sem{\Phi}$ is the set of value assignments satisfying $\Phi$.
\end{definition}
\noindent
The product denotation $\sem{C}$ captures only per-operand restrictions. It ignores
cross-operand relations from the background theory $\Phi$. Therefore, some conflicts
are invisible at the product level and become apparent only after intersecting with
$\sem{\Phi}$. Writing $m$ and $e$ for the metered and elapsed values,
$\odrl{meteredTime}~\odrl{eq}~\textsf{P30D}$ and
$\odrl{elapsedTime}~\odrl{lteq}~\textsf{P10D}$ are consistent on each axis, yet
$\sem{\odrl{meteredTime}}_\ctx\leq\sem{\odrl{elapsedTime}}_\ctx$ forces
$30\leq m\leq e\leq 10$, which is inconsistent.

\noindent
The complexity of conflict detection depends on whether the constraints require simple ordering,
arithmetic between values, or reasoning about periodic structure. These give three tiers. Order
alone ($\Ord$) suffices when each atom bounds one operand by a literal. A difference ($\Diff$) is
needed when an operand is bounded relative to a constrained instant, or when two operands are
related by $\Phi$. Modular reasoning ($\Mod$) is needed when an operand is periodic.
\begin{definition}[Tier]\label{def:tier}
Let $S=C_1\wedge C_2\wedge\Phi$ be the combined system of \cref{def:frame}. A
cross-operand relation of $\Phi$ is \emph{active} when both operands it relates are
constrained by $C_1$ or $C_2$. Then
\[
\tier(S)=
\begin{cases}
\Mod  & \text{some operand of }S\text{ is periodic (\cref{def:ti-rec})},\\
\Diff & \text{otherwise, if }S\text{ contains a difference }\diff(s,t)\text{ over a}\\
      & \text{constrained instant or an active cross-operand relation},\\
\Ord  & \text{otherwise.}
\end{cases}
\]
\end{definition}
\noindent
Each tier has its own decision procedure, each a special case of the next. The order tier
uses a bound comparison, the difference tier a negative-cycle check on a difference graph, and the
modular tier a divisibility test.
\begin{theorem}[Tiered conflict decision]\label{thm:tiered}
Let $S=C_1\wedge C_2\wedge\Phi$. If $\tier(S)=\Ord$, then $S$ is unsatisfiable iff some shared
operand has disjoint bounds (\cref{thm:criterion}), an EPR problem. If $\tier(S)=\Diff$, then $S$
is unsatisfiable iff its difference graph has a negative cycle, in PTIME, which subsumes $\Ord$
since $\Ord\subseteq\Diff$. If $\tier(S)=\Mod$, decide the instant axis and the duration axes and
conjoin: $S_1\cap S_2$ is $\Per(a^{\ast},\operatorname{lcm}(p_1,p_2))$ when
$\gcd(p_1,p_2)\mid(a_2-a_1)$ and $\emptyset$ otherwise (\cref{lem:rec}); with $W$ the feasible
$\now$-interval cut out by the \odrl{dateTime} bounds and the difference system, the instant axis
conflicts iff $S_1\cap S_2=\emptyset$, $W=\emptyset$, or $W\cap(S_1\cap S_2)=\emptyset$
(\cref{lem:windowed-rec}), and $S$ is unsatisfiable iff the instant axis conflicts or the duration
difference system has a negative cycle. Each sub-check is in Presburger arithmetic. The decision is
sound and complete relative to $\Phi$: every conflict entailed by $C_1\wedge C_2\wedge\Phi$ is
detected; $\Phi$ (\cref{def:frame}) fixes the domain relations (durations nonnegative, metered at most elapsed, and delay at most elapsed), so a conflict resting on a relation outside $\Phi$ is not guaranteed. The proof is in \cref{app:proofs}.
\end{theorem}
\noindent
\begin{corollary}[Product exactness]\label{cor:product-exact}
If $\tier(S)=\Ord$, the difference graph of \cref{thm:tiered} has no cross-operand edge, so it
decomposes operand by operand and conflict is decided by the per-operand criterion
(\cref{thm:criterion}). The product denotation is then exact, while in $\Diff$ and $\Mod$ it
over-approximates.
\end{corollary}
\noindent
In \cref{ex:bsb} the \odrl{dateTime} and \odrl{elapsedTime} pairs are $\Ord$, the
\odrl{timeInterval} pair is $\Mod$, and the metered-against-elapsed extension of \cref{def:frame}
is $\Diff$, where $30\leq m\leq e\leq 10$ is a negative cycle.
\noindent
Full metric temporal logic over dense time is undecidable~\cite{ouaknine2005decidability}, so we stay in decidable fragments: $\Diff$ in difference logic and $\Mod$ in Presburger arithmetic over integer period indices, with the time model dense throughout.
\vspace{-10pt}
\section{Conflict Detection, Composition, and Runtime}
\label{sec:conflict}
\vspace{-10pt}
The three verdicts form a totally ordered algebra in which conjunction and disjunction are the
strong Kleene meet and join.
\begin{definition}[Verdict algebra]\label{def:verdict-algebra}
$\Verdicts=\{\Conflict,\Unknown,\Compatible\}$ is totally ordered by
$\Conflict\sqsubset\Unknown\sqsubset\Compatible$. Conjunction $\wedge_K$ is $\min$ and disjunction
$\vee_K$ is $\max$ under $\sqsubset$.
\end{definition}
\noindent
$\Conflict$ and $\Compatible$ are determinate; $\Unknown$ is epistemic.
\begin{definition}[Per-operand verdict]\label{def:operand-verdict}
For an operand $\ell$ constrained by at least one of $C_1,C_2$,
\[
V_\ell=
\begin{cases}
\Conflict   & \text{both constrain }\ell\text{ and their denotations on }\ell\text{ are disjoint},\\
\Compatible & \text{both constrain }\ell\text{ and their denotations on }\ell\text{ intersect},\\
\Unknown    & \text{exactly one of }C_1,C_2\text{ constrains }\ell.
\end{cases}
\]
Here $\ell$ ranges over the duration operands \odrl{elapsedTime}, \odrl{meteredTime},
\odrl{delayPeriod}; the denotations on $\ell$ are the normalised intervals
$I_\ell(C_1),I_\ell(C_2)$ and disjointness is decided by \cref{thm:criterion}. The
instant operands \odrl{dateTime} and \odrl{timeInterval} are handled jointly by the
instant-axis verdict (\cref{def:instant-verdict}).
\end{definition}
\noindent
\begin{definition}[Instant-axis verdict]\label{def:instant-verdict}
Say a policy \emph{constrains the instant axis} when it has a \odrl{dateTime} or a
\odrl{timeInterval} constraint. With $N(\cdot)$ as in \cref{def:now-denotation},
\[
V_{\mathsf{Inst}}=
\begin{cases}
\Conflict   & \text{both }C_1,C_2\text{ constrain the instant axis and }N(C_1)\cap N(C_2)=\emptyset,\\
\Compatible & \text{both constrain the instant axis and }N(C_1)\cap N(C_2)\neq\emptyset,\\
\Unknown    & \text{exactly one of }C_1,C_2\text{ constrains the instant axis.}
\end{cases}
\]
\end{definition}
\begin{definition}[Product (policy) verdict]\label{def:product-verdict}
For $C_1,C_2$, the \emph{policy verdict} is
\[
\prodV(C_1,C_2)=\min\bigl(V_{\mathsf{Inst}},\ \min_{\ell}V_\ell\bigr),
\]
the minimum of the instant-axis verdict (\cref{def:instant-verdict}) and the
per-operand verdicts (\cref{def:operand-verdict}) over the duration operands
constrained by at least one of $C_1,C_2$. \odrl{dateTime} and \odrl{timeInterval}
contribute only through $V_{\mathsf{Inst}}$, not as separate product axes.
\end{definition}
\noindent
On operands both policies constrain, the verdict is determinate and, in the $\Ord$
tier, agrees with the conflict of \cref{def:frame}. On an operand only one policy
constrains it is $\Unknown$, a case \cref{def:frame} does not separate from
compatibility. The $\Diff$ and $\Mod$ tiers use \cref{thm:tiered}.
\noindent In \cref{ex:bsb}, $V_{\mathsf{Inst}}=\Conflict$ (the offer's window-and-recurrence set
and the request's share no instant), $V_{\odrl{elapsedTime}}=\Compatible$,
$V_{\odrl{delayPeriod}}=\Conflict$, and $V_{\odrl{meteredTime}}=\Unknown$ (only the offer
constrains it), so $\prodV=\Conflict$.
\begin{proposition}[Monotonicity]\label{prop:monotone}
If $C_1$ and $C_2$ both constrain $\ell$, $I_\ell(C_1')\subseteq I_\ell(C_1)$, and
the other operands are unchanged, then $\prodV(C_1',C_2)\sqsubseteq\prodV(C_1,C_2)$.
\end{proposition}
\begin{proof}
On $\ell$, shrinking to $I_\ell(C_1')\subseteq I_\ell(C_1)$ keeps a disjoint pair
disjoint and can only turn an intersecting pair disjoint, so $V_\ell'\sqsubseteq
V_\ell$. The other operands are unchanged, so
$\prodV(C_1',C_2)=\min_\ell V_\ell'\sqsubseteq\min_\ell V_\ell=\prodV(C_1,C_2)$.
\end{proof}
\noindent
For instance, an offer permitting \odrl{dateTime} until \textsf{2026-12-31} is
\Compatible\ with a request fixing \textsf{2026-06-01}; tightening the deadline to
\textsf{2026-03-01} excludes that date, giving \Conflict.
\noindent
An $\Unknown$ verdict is undetermined: completing the silent policy, that is, adding a
constraint for the operand it leaves open, can make the pair $\Compatible$ or $\Conflict$.
\begin{definition}[Completion]\label{def:completion}
A \emph{completion} of $C_1,C_2$ is a pair $(\hat C_1,\hat C_2)$ with
$C_i\subseteq\hat C_i$ that adds, for every operand $\ell$ unconstrained by $C_i$,
exactly one admissible constraint targeting $\ell$.
\end{definition}
\begin{theorem}[Unknown soundness]\label{thm:unknown-sound}
$\prodV(C_1,C_2)=\Unknown$ iff no operand constrained by both policies yields
$\Conflict$ and some operand is constrained by exactly one of them. In that case a
completion with $\prodV=\Compatible$ exists, and, provided some operand $\ell$
constrained by exactly one policy $C_j$ has $I_\ell(C_j)\neq D_\ell$, a completion
with $\prodV=\Conflict$ exists.
\end{theorem} \noindent
The proof is in \cref{app:proofs}. For instance, when only one policy bounds an operand, completing
the other with an overlapping bound on it yields \Compatible\ and with a disjoint bound \Conflict.
\vspace{-10pt}
\subsection{Logical Composition}
\label{sec:composition}
The logical operands \odrl{and}, \odrl{or}, and \odrl{xone} combine constraints evaluated at a
single exercise, by intersection, union, and exactly-one-of on their denotations. \odrl{andSequence}
sequences constraints over a trace rather than at a single exercise, and is treated separately.
\begin{definition}[Logical composition]\label{def:composition}
Over one exercise,
$\sem{\odrl{and}(C_1,\dots,C_n)}=\bigcap_i\sem{C_i}$,
$\sem{\odrl{or}(C_1,\dots,C_n)}=\bigcup_i\sem{C_i}$, and
$\sem{\odrl{xone}(C_1,\dots,C_n)}
=\bigcup_i\big(\sem{C_i}\setminus\bigcup_{j\neq i}\sem{C_j}\big)$.
\end{definition}
\noindent
A disjunctive policy denotes a union of products, so its verdict against another
policy is determined by the verdicts of all branch pairs.
\begin{definition}[Branch]\label{def:branch}
A \emph{branch} is a disjunct of a policy's \odrl{or} or \odrl{xone} constraint, a
constraint or conjunction, whose product verdict against another branch is given by
\cref{def:product-verdict}. A \emph{branch pair} $(B_i,B_j)$ takes one branch from
each policy.
\end{definition}
\noindent
\begin{definition}[Disjunction and exclusive-disjunction verdicts]\label{def:or-verdict}
$\verdict_{\odrl{or}}$ is \Compatible\ if $\prodV(B_i,B_j)=\Compatible$ for some
pair, \Conflict\ if every pair is \Conflict, and \Unknown\ otherwise. For
\odrl{xone} the pair verdicts give only the sound rule $\verdict_{\odrl{xone}}
=\Conflict$ when every pair is \Conflict; otherwise the verdict is decided on the
\odrl{xone} denotation (\cref{def:composition}): \Compatible\ when
$\sem{\odrl{xone}(C_1,\dots,C_n)}\cap\sem{\odrl{xone}(C_1',\dots,C_m')}\neq\emptyset$, and \Conflict\ when it is empty. Counting \Compatible\ pairs decides neither
direction, since a satisfying point must lie in exactly one branch of each side.
\end{definition}
\noindent
For each combinator a \Conflict\ verdict reduces to disjointness of the combined
denotations, since intersection distributes over the unions that \odrl{or} and
\odrl{xone} introduce.
\begin{theorem}[Composition soundness]\label{thm:composition-soundness}
A \Conflict\ verdict under \odrl{and}, \odrl{or}, or \odrl{xone} implies that the
two policies' denotations are disjoint.
\end{theorem}
\noindent The proof is in \cref{app:proofs}.
\noindent
The product verdict (\cref{def:product-verdict}) is sound but not complete: it ranges
over operands independently and so omits the cross-operand relations of $\Phi$, and a
conflict arising only through such a relation is missed. The $\Diff$ tier of \cref{thm:tiered} catches exactly the cross-operand conflicts, and the $\Mod$ tier the periodic ones.
\vspace{-10pt}
\subsection{Sequences}
\label{sec:sequences}
Sequencing is evaluated over a trace, the timestamped states of a single execution.
\begin{definition}[Trace]\label{def:trace}
A \emph{trace} is a sequence $\Trace=\langle S_1,\dots,S_z\rangle$ with strictly
increasing timestamps $t_{S_1}<\cdots<t_{S_z}$, each $S_i$ inducing a context
$\ctx_i$ with $\ctx_i.\now=t_{S_i}$.
\end{definition}

\noindent
\odrl{andSequence} requires its constraints to be satisfied in order along the trace,
with each delay measured from the previous satisfaction. It is the only logical operand in which a later constraint depends on an earlier instant.

\begin{definition}[\odrl{andSequence}]\label{def:andseq}
$\odrl{andSequence}(C_1,\dots,C_n)$ holds on $\Trace$ iff there are indices
$1\leq j_1\leq\cdots\leq j_n\leq z$ such that each $C_k$ is satisfied at $S_{j_k}$ in
the context $\ctx_{j_k}$ whose trigger is bound to the preceding satisfaction instant,
$\ctx_{j_k}.\trig=t_{S_{j_{k-1}}}$, with $\trig$ free for $k=1$. It is
\emph{unresolvable} on $\Trace$ when no such non-decreasing tuple exists.
\end{definition}
\noindent
The index order $j_1\leq\cdots\leq j_n$ is non-strict, so two consecutive constraints may be
satisfied at the same state. A step forces strict separation $j_{k-1}<j_k$ exactly when it
imposes a positive delay.
\begin{lemma}[Delay-induced separation]\label{lem:andseq-strict}
If $C_k$ depends on the trigger $\ctx_{j_k}.\trig=t_{S_{j_{k-1}}}$ and forces
$\diff(\ctx_{j_k}.\trig,\ctx_{j_k}.\now)>0$ (e.g.\ $\odrl{delayPeriod}~\odrl{gt}~r$
with $r\ge0$, $\odrl{delayPeriod}~\odrl{gteq}~r$ with $r>0$, or
$\odrl{delayPeriod}~\odrl{eq}~r$ with $r>0$), then every satisfying tuple has
$t_{S_{j_k}}>t_{S_{j_{k-1}}}$, hence $j_{k-1}<j_k$.
\end{lemma}
\begin{proof}
At such a step $\diff(\ctx_{j_k}.\trig,\ctx_{j_k}.\now)=t_{S_{j_k}}-t_{S_{j_{k-1}}}>0$,
so $t_{S_{j_k}}>t_{S_{j_{k-1}}}$; since $\Trace$ has strictly increasing timestamps,
$j_{k-1}<j_k$.
\end{proof}
\noindent
Strict separation thus holds exactly at steps carrying a positive delay, and all other steps
order non-strictly. For instance,
$\odrl{andSequence}(\odrl{dateTime}~\odrl{gteq}~\textsf{2026-01-01},\
\odrl{delayPeriod}~\odrl{gteq}~\textsf{P1D})$ requires the second constraint to be satisfied at
least \textsf{P1D} after the first. \noindent
\cref{prop:collapse} gives the single-state case; otherwise a sequence is decided over the trace as a difference system.
\begin{proposition}[Collapse to \odrl{and}]\label{prop:collapse}
If no $C_k$ depends on $\trig$, then for every state $S$, $\odrl{andSequence}(C_1,\dots,C_n)$
holds on the single-state trace $\langle S\rangle$ iff $\odrl{and}(C_1,\dots,C_n)$ holds at $S$,
equivalently when the two have the same per-state denotation $\bigcap_k\sem{C_k}$.
\end{proposition}
\begin{proof}
With no constraint depending on $\trig$, each $\sem{C_k}_\ctx$ is determined by the state. On
$\langle S\rangle$ the constant tuple $j_1=\cdots=j_n=1$ is the only non-decreasing choice and is
admissible, so the sequence holds iff every $C_k$ holds at $S$, that is iff
$S\in\bigcap_k\sem{C_k}$.
\end{proof}
\noindent
The only operand whose value depends on $\ctx.\trig$ is \odrl{delayPeriod}. Under
\odrl{andSequence} its value $\diff(\ctx.\trig,\ctx.\now)$ is the duration from the trigger to
$\ctx.\now$, the one case in which positional order and metric delay are decided together.
Encoding the trigger as a timepoint and each delay as a difference classifies this system in the
$\Diff$ tier of \cref{thm:tiered}, where the sequence is unresolvable exactly when the difference
system is infeasible, a negative cycle. The other duration operands take their values from
context-fixed instants and introduce no cross-operand difference.
\vspace{-10pt}
\subsection{Static and Runtime Conflict}
\label{sec:static-runtime}
Static and runtime conflict are the same question under different quantification. A policy
pair is in \emph{static conflict} when the two are jointly unsatisfiable in the sense of
\cref{def:frame}, the context left symbolic: the negotiation-time check, before any access,
decided by \cref{thm:tiered}. In \cref{ex:bsb} the \odrl{dateTime} disjointness is a static
\Conflict\ for every exercise. A \emph{runtime conflict} holds when a concrete ground
context, or a trace under \odrl{andSequence}, fails a policy at enforcement, the duration
operands computed from actual timestamps and the metered total rather than searched; it
explores no executions, only checking observed usage against the constraints.
\begin{definition}[Ground evaluation]\label{def:runtime-eval}
A \emph{ground context} $\hat\ctx$ fixes concrete $\now,\trig\in\timeT$ and the access-begin $a_0$ of \cref{def:frame}. Over usage segments
$[s_1,e_1],\dots,[s_z,e_z]$, the duration operands are computed as
$\sem{\odrl{elapsedTime}}_{\hat\ctx}=\diff(a_0,\now)$,
$\sem{\odrl{delayPeriod}}_{\hat\ctx}=\diff(\trig,\now)$, and
$\sem{\odrl{meteredTime}}_{\hat\ctx}=\sum_i\diff(s_i,e_i)$, and $\hat\ctx\models C$ holds iff
$(\sem{\ell}_{\hat\ctx})_{\ell}\in\sem{C}$.
\end{definition}
\noindent
This evaluation is a ground computation. Its only arithmetic is $+$ over $\durD$, and it performs
no search. For instance, with $a_0=\textsf{2027-06-01}$ and $\now=\textsf{2027-06-09}$,
$\sem{\odrl{elapsedTime}}_{\hat\ctx}=\textsf{P8D}$, satisfying
$\odrl{elapsedTime}~\odrl{lteq}~\textsf{P10D}$.
\begin{definition}[Runtime conflict]\label{def:runtime-conflict}
A ground context $\hat\ctx$ is in \emph{runtime conflict} with $(C_1,C_2)$ iff
$\hat\ctx\not\models C_1$ or $\hat\ctx\not\models C_2$ (\cref{def:runtime-eval}).
\end{definition}
\noindent
A trace whose metered usage exceeds the requested cap is in runtime conflict even
where the static product was \Compatible.

\begin{definition}[Realizability]\label{def:realizability}
A policy $C$ is \emph{realizable} iff $\exists\Trace.\ \Trace\models C$.
\end{definition}

\noindent
Under \odrl{andSequence} with a trigger-dependent \odrl{delayPeriod}, realizability is
an $\exists$ over the trigger instants subject to the sequence order and the bound
$\diff(\trig,\now)\bowtie r$, a simple temporal network (STN) emptiness check; an
unresolvable sequence (\cref{def:andseq}) gives an unrealizable $C$. For instance, an
\odrl{andSequence} demanding $\odrl{delayPeriod}~\odrl{gteq}~\textsf{P2D}$ within a
one-day window is~unrealizable.
\noindent
\begin{definition}[Trace conflict]\label{def:trace-conflict}
For an \odrl{andSequence}, satisfaction is trace-based and
\cref{def:frame} does not apply. $C_1,C_2$ are in \emph{trace conflict} when no trace
satisfies both, equivalently when the joint difference system is infeasible
(\cref{thm:tiered}, $\Diff$ tier).
\end{definition}
\noindent
Static conflict and runtime conflict relate by refinement: static conflict is
conclusive, runtime conflict is per-context.
\begin{theorem}[Static--runtime bridge]\label{thm:static-runtime}
\emph{Soundness.} If $C_1,C_2$ are in static conflict, no exercise satisfies both, by product
emptiness (\cref{def:frame}) for single-state policies and trace conflict
(\cref{def:trace-conflict}) for an  \odrl{andSequence}, so a static \Conflict\ is
runtime-safe. \emph{Refinement, not reversal.} A static non-conflict does not make every exercise
compliant, since a particular $\hat\ctx$ may fail one policy, for instance actual metered usage
exceeding a cap the two policies can jointly meet; by soundness, no $\hat\ctx$ overturns a static
\Conflict. \emph{Theory.} Static conflict is the tiered decision (\cref{thm:tiered}), ground
evaluation computes only $\diff$ and the metered sum, and realizability is the one runtime task
that needs search, an $\exists$ over a trace in the $\Diff$ tier, escalating to $\Mod$ only with a
periodic operand.
\end{theorem}
\noindent
\begin{proof}
\emph{(1)} A satisfying $\hat\ctx$ lies in $\sem{C_1}\cap\sem{C_2}\cap\sem{\Phi}$; a
satisfying trace solves the joint difference system; both contradict emptiness.
\emph{(2)} A static non-conflict has a satisfying exercise, but a given $\hat\ctx$
need not be one and may fail $C_1$ or $C_2$; soundness of~(1) forbids the reverse.
\emph{(3)} By the tier classification and the fact that ground evaluation is
substitution.
\end{proof}
\begin{corollary}[Positional and metric in one operand]\label{cor:bridge}
Under \odrl{andSequence}, \odrl{delayPeriod} carries both a positional component, the
chosen trigger instant $\trig$, and a metric component, the duration comparison
$\diff(\trig,\now)$. It therefore cannot be encoded by order alone and lies in the $\Diff$ tier.
\end{corollary}
\vspace{-10pt}
\subsection{Entailment and Policy Comparison}
\label{sec:refinement}
One constraint entails another when every execution satisfying the first also satisfies the
second. Entailment orders constraints by restrictiveness and is finer than compatibility, which
only checks whether two policies are jointly satisfiable.
\begin{definition}[Entailment]\label{def:refinement}
$C_1\refines C_2$ iff $\sem{C_1}\subseteq\sem{C_2}$.
\end{definition}
\noindent
On a totally ordered domain, entailment of atomic constraints reduces to a comparison of right
operands, for example $\odrl{elapsedTime}~\odrl{lteq}~\textsf{P5D}\refines
\odrl{elapsedTime}~\odrl{lteq}~\textsf{P30D}$.
\begin{lemma}[Syntactic entailment]\label{lem:refine-syntax}
On $\Inst$, $(\odrl{dateTime},\odrl{lt},r_1)\refines(\odrl{dateTime},\odrl{lt},r_2)$ iff
$r_1\leq_{\timeT}r_2$, and dually for $\odrl{gt}$. On $\Dur$, for
$\ell\in\{\odrl{elapsedTime},\odrl{meteredTime}\}$,
$(\ell,\odrl{lteq},r_1)\refines(\ell,\odrl{lteq},r_2)$ iff $r_1\leq_{\durD}r_2$, and for
$\odrl{delayPeriod}$,
$(\odrl{delayPeriod},\odrl{gteq},r_1)\refines(\odrl{delayPeriod},\odrl{gteq},r_2)$ iff
$r_1\geq_{\durD}r_2$. Hence $\refines$ between atomic constraints is decidable by one comparison
of right operands.
\end{lemma}
\begin{proof}
Each clause computes interval containment on a totally ordered domain, which reduces to a single
comparison.
\end{proof}
\noindent
Deciding $\refines$ decides policy comparison at the constraint level. For atomic constraints
this is the single comparison of \cref{lem:refine-syntax}, which \cref{thm:tiered} classifies in
the $\Ord$ (EPR) tier. Entailment of sequenced policies does not reduce to a single comparison,
since it involves the \odrl{delayPeriod} difference. 
\section{Evaluation}
\label{sec:eval}
\begin{table}[t]
\centering
\footnotesize
\caption{The 72 conflict-detection problems by the solver that decides them (\SI{20}{\second}, single core).}
\label{tab:eval}
\begin{tabular}{@{}lp{4.6cm}rl@{}}
\toprule
Group & Categories & $n$ & Decided by \\
\midrule
Order (first-order)        & single-operand, conflict criterion, \odrl{or}, \odrl{xone}, unknown, refinement, completion, composition, capstone (verdict vector), realizability (window), sort ablation & 48 & all four \\
Arithmetic (Vampire + SMT) & cross-operand, capstone, runtime, periodic (compatible) & 14 & Vampire, Z3, cvc5 \\
Arithmetic (SMT only)      & periodic (conflict), sequence, realizability (network) & 10 & Z3, cvc5 \\
\midrule
Total                      & & 72 & \\
\bottomrule
\end{tabular}
\end{table}
\noindent\textbf{Setup.} We compile the semantics into a benchmark of 72
conflict-detection problems in 15 categories, each emitted as a TPTP~\cite{sutcliffe2017tptp} and an
SMT-LIB~\cite{BarFT-RR-25} problem with its expected verdict,\footnote{Each problem stores its
expected SZS status: a conflict as the conjecture $\forall x.\,\lnot(m_1(x)\land m_2(x))$ (status
\textsf{Theorem}), a compatibility as a witness $\exists x.\,(m_1(x)\land m_2(x))$, and a
three-valued case as an equation over the verdict algebra; the SMT-LIB side asserts the constraints,
and the harness maps \textsf{unsat}/\textsf{sat} onto these, counting a definite disagreement as a
wrong verdict and a \textsf{Timeout} as undecided.} together with its source ODRL policy in Turtle,
and decide them with the first-order provers Vampire~5.0.1~\cite{kovacs2013vampire},
E~3.3.2~\cite{schulz2019eprover} (first-order only), Z3~4.8.12~\cite{demoura2008z3}, and
cvc5~1.3.4~\cite{barbosa2022cvc5}, under a \SI{20}{\second} timeout on a single core. The duration, recurrence, cross-operand, and sequence categories are outside the scope of existing
ODRL engines by construction: they read \odrl{dateTime} alone, and the runtime evaluators among them
score one policy against observed usage rather than checking an offer against a request. 

\noindent\textbf{Results.} Every problem is decided by at least one reasoner and none returns a
wrong verdict; and since the two encodings are separate and the four provers agree, the verdicts
are properties of the semantics rather than of an encoding. Several problems isolate a single
mechanism: an unsorted reading reports a spurious conflict (ODRL869) that vanishes once each
operand is sorted (ODRL870); dropping $\mathit{metered}\le\mathit{elapsed}$ turns a cross-operand
conflict into compatibility (ODRL821); the same branches give Compatible under \odrl{or} but
Conflict under \odrl{xone} (ODRL846); two statically compatible caps conflict once runtime usage
exceeds the tighter one (ODRL840); a silent operand yields \emph{Unknown} that completion resolves
upward or downward (ODRL856, ODRL857);  the capstone problem, the running example, folds its four per-operand verdicts (Conflict,
Compatible, Unknown, Conflict) into one overall Conflict (ODRL863); and a windowed recurrence conflicts where two identical schedules
share no occurrence in the joint \odrl{dateTime} window (ODRL871). By backend (Q3,
\cref{tab:eval}), E decides exactly the 48 first-order problems, since the arithmetic fragments are
TFF over \texttt{\$int}, which it does not support; Vampire decides 62, those 48 plus 14 arithmetic
problems (the cross-operand, capstone, and runtime conflicts, each a single difference or
equality-pinned bound, and the two periodic-compatible cases); and Z3 and cvc5 decide all 72. The
10 Vampire misses are the periodic conflicts, whose disjointness is a divisibility argument, and
the sequence and network-realizability cases, whose satisfaction is the feasibility of a difference
system over a whole trace. \textbf{Discussion.} The benchmark confirms the semantics is faithfully mechanised, since
the verdicts are exactly those computed, and maps the reasoning each construct needs to a backend:
order conflicts fall to any prover, the cross-operand and bounded-arithmetic conflicts to Vampire,
and divisibility and trace-feasibility to the SMT solvers, which informs the choice of engine.
\section{Related Work}
\label{sec:related}
ODRL has been formalised in several ways: Steyskal and Polleres give an early
semantics over an abstract syntax~\cite{steyskal2015}, De~Vos et al.\ encode
policies in Answer Set Programming for compliance checking~\cite{de2019odrl},
and a W3C Community Group report is drafting a formal semantics for the
standard~\cite{odrlfs2026}. ODRL's own \odrl{conflict} property~\cite{odrl-model} and
GUCON~\cite{akaichi2023gucon} resolve clashes between permissions and
prohibitions; the joint satisfiability of two permissions studied here is a
separate question. The closest formal works are Salas et al.\ and Bonatti et al. Salas et al.\
compare a provider and a requester policy at negotiation time by overlap,
containment, and equivalence, later adding a size bound through
normalisation~\cite{salas2025evaluationcomparisonsemanticsodrl,salas2026normalisation}.
Bonatti et al.\ give a declarative semantics interpreting each operator against
its domain's ordering, and characterise comparison as trace containment, a
correctness criterion with no decision procedure~\cite{bonatti2025towards}.
Both treat time through \odrl{dateTime} alone, as a sortless ordered attribute,
so neither distinguishes instants from durations nor models duration
arithmetic, recurrence, or cross-operand relations, and both set
\odrl{andSequence} aside as indistinguishable from \odrl{and}. Our work
differs in kind: it distinguishes instants from durations, places each conflict
in a decidable tier, and gives \odrl{andSequence} a semantics. Axis-Aligned
Profile~\cite{mustafa2026axisalignedsemanticsodrlresolving} resolves the analogous spatial ambiguity by axis
decomposition and states the same semantics should extend to time, but does not
treat it; the temporal case is a sort problem, not an axis one, needing
stratification rather than decomposition. On the implementation side, ODRE, the ODRL Evaluator, and the EDC
connector~\cite{EDC2026} evaluate or match policies at runtime, returning a usage decision and a compliance report
respectively~\cite{cimmino2025open,slabbinck2025interoperable}, and Cano-Benito
et al.\ extend ODRE with OWL-Time relations~\cite{cano2024towards}. All compute
only \odrl{dateTime}, and the Evaluator treats \odrl{andSequence} as plain
\odrl{and}. None compares two policies: joint satisfiability of an offer and a
request over all temporal operands. The tiers reuse established procedures: Allen-style interval
reasoning~\cite{allen1983maintaining} for the order tier, the negative-cycle test of
simple temporal networks~\cite{dechter1991temporal} for the difference tier, and a
gcd-and-divisibility test in Presburger arithmetic~\cite{presburger} for the
modular tier. The boundary is deliberate: metric deadlines over dense time fall
under metric temporal logic, which is undecidable~\cite{ouaknine2005decidability}.
The contribution here is the decidability characterisation and the tiered
reduction, each tier then discharged by a decision procedure.

\section{Conclusions and Future Work}
\label{sec:conclusion}
We gave a sort-stratified semantics for ODRL's temporal constraints that types each operand as an
instant or a duration and reads the comparison operators per sort, removing the ambiguity that made
conflict detection unsound. Conflict between an offer and a request then reduces to interval
comparison under a background theory and a three-valued verdict, covering duration arithmetic,
recurrence, cross-operand relations, and \odrl{andSequence}. Each conflict falls into one of three
decidable tiers, negotiation-time and runtime conflict are tied by a soundness bridge, and a
benchmark on which four reasoners agree confirms the semantics is faithfully mechanised. \textbf{Future work.} The sort discipline extends to ODRL's other temporal operands,
positions within a media stream; the static semantics could drive an enforcement
engine that computes the duration and recurrence operands current ODRL engines omit; and a
machine-checked development(Isabelle/HOL) of the soundness results is in progress.
\newpage
\bibliographystyle{plain}
\bibliography{easychair}

\newpage
\appendix
\section{Proofs}\label{app:proofs}

\begin{proof}[Proof of \cref{thm:product-expressibility}] The shapes $[v,v]$, $(\bot,v]$, $[v,\top)$, $(\bot,v)$, $(v,\top)$ are denoted by
\odrl{eq}, \odrl{lteq}, \odrl{gteq}, \odrl{lt}, \odrl{gt}; the full domain by no
constraint; a two-sided interval by one lower and one upper constraint, which under \cref{def:admissible} is realizable only on \odrl{dateTime}; the duration operands admit one-sided or point intervals only. Each operator
is admissible for $\ell_k$ by hypothesis. Conjoining the per-operand counterparts
gives $C$. \end{proof}

\begin{proof}[Proof of \cref{thm:tiered}]
\emph{Ord.} Every atom bounds one operand by a literal, so on each shared operand
$S$ is unsatisfiable iff the two intervals are disjoint, decided by \cref{thm:criterion};
no cross-operand edge arises, an EPR problem.
\emph{Diff.} Encode each instant as a timepoint, a trigger-bound \odrl{delayPeriod} as
$\now-\trig$, and each accumulating duration as a value variable bounded by $\Phi$; every
atom is then a difference $x-y\bowtie c$, so $S$ is unsatisfiable iff its difference graph
has a negative cycle, in PTIME, which subsumes Ord since $\Ord\subseteq\Diff$.
\emph{Mod.} The duration axes are the difference system of the Diff case. On the instant
axis, by \cref{def:now-denotation} a satisfying $\now$ lies in
$(W_1\cap W_2)\cap(S_1\cap S_2)=W\cap(S_1\cap S_2)$; this is empty iff $S_1\cap S_2$ is empty
(incompatible cycles, \cref{lem:rec}), or $W$ is empty (infeasible difference/interval
system), or no occurrence of $S_1\cap S_2$ lies in $W$ (\cref{lem:windowed-rec}). A trace
satisfies both policies iff it meets every axis, so $S$ is satisfiable iff the duration
system has no negative cycle and the instant axis is nonempty; negating gives the stated
condition.
\emph{Soundness and relative completeness.} Each case decides $C_1\wedge C_2\wedge\Phi$
exactly, so a reported $\Conflict$ is a real conflict, and every conflict entailed by
$C_1\wedge C_2\wedge\Phi$ is reported.
\end{proof}

\begin{proof}[Proof of \cref{thm:unknown-sound}]
$\prodV=\min(V_{\mathsf{Inst}},\min_\ell V_\ell)=\Unknown$. For $\Compatible$, complete each operand a policy leaves
unconstrained with an $\odrl{eq}$ value in the other policy's interval there,
sharing a value where neither constrains it; every $V_\ell=\Compatible$, since
operands constrained by both are $\Compatible$ by hypothesis. For $\Conflict$, take $\ell$ constrained by exactly one policy $C_j$ with
$I_\ell(C_j)\neq D_\ell$, pick $v\in D_\ell\setminus I_\ell(C_j)$, add
$(\ell,\odrl{eq},v)$ to the policy lacking $\ell$, and complete every other
unconstrained operand as in the $\Compatible$ case; then
$\{v\}\cap I_\ell(C_j)=\emptyset$, so $V_\ell=\Conflict$ (\cref{thm:criterion}) and
$\prodV=\Conflict$ regardless of the others. The instant axis is completed the same way
through $N(\cdot)$ (\cref{def:now-denotation}): when only $C_j$ constrains it, add a
$(\odrl{dateTime},\odrl{eq},t)$ to the other policy with $t\in N(C_j)$ for $\Compatible$,
and, when $N(C_j)\neq\timeT$, with $t\notin N(C_j)$ (a non-occurrence of the recurrence if
$C_j$ constrains only \odrl{timeInterval}) for $\Conflict$; then $V_{\mathsf{Inst}}$ takes
the chosen verdict, so the $\Conflict$ witness may be this axis (with $N(C_j)\neq\timeT$) in
place of a duration operand.
\end{proof}

\begin{proof}[Proof of \cref{thm:composition-soundness}] For conjunction, some $V_\ell=\Conflict$ gives $I_\ell(C_1)\cap I_\ell(C_2)=\emptyset$, hence
$\sem{C_1}\cap\sem{C_2}=\emptyset$ by \cref{def:product-denotation}, since an empty factor empties
the product. For disjunction, every pair conflicts, so
$\big(\bigcup_i\sem{B_i}\big)\cap\big(\bigcup_j\sem{B_j}\big)
=\bigcup_{i,j}\big(\sem{B_i}\cap\sem{B_j}\big)=\emptyset$. For exclusive disjunction,
$\verdict_{\odrl{xone}}=\Conflict$ means
$\sem{\odrl{xone}(C_1,\dots,C_n)}\cap\sem{\odrl{xone}(C_1',\dots,C_m')}=\emptyset$, read directly
or, when every pair conflicts, because the \odrl{xone} sets lie inside the branch unions. Either
way the policies are disjoint. \end{proof}
\end{document}

%% file: fig-bsb-denotations.tex
%
\tikzset{
  >={Stealth[length=3.3pt]},
  mx/.style={gray!55,line width=0.45pt},                      
  mo/.style={cOfferText,line width=1.5pt},                    
  mr/.style={cRequestText,line width=1.5pt},                  
  co/.style={cOfferText,line width=1.1pt},                    
  cr/.style={cRequestText,line width=1.1pt},                  
  oep/.style={cOfferText,fill=cOfferText},                    
  rep/.style={cRequestText,fill=cRequestText},                
  ohole/.style={cOfferText,fill=white,line width=0.9pt},      
  mol/.style={font=\scriptsize,cOfferText,anchor=south,inner sep=0.6pt},
  mrl/.style={font=\scriptsize,cRequestText,anchor=north,inner sep=0.6pt},
}
\ifdefined\dotr\else\newlength{\dotr}\fi  \setlength{\dotr}{1.8pt}
\ifdefined\dnu\else\newlength{\dnu}\fi    
\colorlet{cOverlapBg}{green!15}           
\newcolumntype{C}{>{\centering\arraybackslash}X}  
\newcommand{\lrule}{\arrayrulecolor{black!18}\specialrule{0.35pt}{2.5pt}{2.5pt}\arrayrulecolor{black!40}}
\newcommand{\dnfit}{\setlength{\dnu}{\dimexpr\linewidth*10/39\relax}}
\newcommand{\dnDate}{\dnfit\begin{tikzpicture}[x=\dnu,y=1cm,baseline=-2.5pt]
  \draw[mx,->](0,0)--(3.7,0);
  \draw[mo,->](1.5,0.15)--(0.2,0.15);\draw[ohole](1.5,0.15)circle(\dotr);
  \node[mol]at(0.85,0.24){$<$\,\textsf{2026-12-31}};
  \draw[rep](2.85,-0.15)circle(\dotr);\node[mrl]at(2.8,-0.24){$=$\,\textsf{2027-06-01}};
\end{tikzpicture}}
\newcommand{\dnElapsed}{\dnfit\begin{tikzpicture}[x=\dnu,y=1cm,baseline=-2.5pt]
  \fill[cOverlapBg](0.2,-0.2)rectangle(1.05,0.2);
  \draw[mx,->](0,0)--(3.7,0);
  \draw[mo](0.2,0.15)--(3.0,0.15);\draw[oep](0.2,0.15)circle(\dotr);\draw[oep](3.0,0.15)circle(\dotr);
  \node[mol]at(1.6,0.24){$\le$\,\textsf{P30D}};
  \draw[mr](0.2,-0.15)--(1.05,-0.15);\draw[rep](0.2,-0.15)circle(\dotr);\draw[rep](1.05,-0.15)circle(\dotr);
  \node[mrl]at(0.62,-0.24){$\le$\,\textsf{P10D}};
\end{tikzpicture}}
\newcommand{\dnMetered}{\dnfit\begin{tikzpicture}[x=\dnu,y=1cm,baseline=-2.5pt]
  \draw[mx,->](0,0)--(3.7,0);
  \draw[oep](3.0,0.15)circle(\dotr);\node[mol]at(2.95,0.24){$=$\,\textsf{P30D}};
\node[font=\scriptsize\itshape,cRequestText,anchor=north,inner sep=0.6pt]at(0.62,-0.20){(none)};
\end{tikzpicture}}
\newcommand{\dnDelay}{\dnfit\begin{tikzpicture}[x=\dnu,y=1cm,baseline=-2.5pt]
  \draw[mx,->](0,0)--(3.7,0);
  \draw[mo,->](0.95,0.15)--(3.4,0.15);\draw[oep](0.95,0.15)circle(\dotr);
  \node[mol]at(2.0,0.24){$\ge$\,\textsf{P1D}};
  \draw[rep](0.2,-0.15)circle(\dotr);\node[mrl,anchor=north west]at(0.08,-0.24){$=$\,\textsf{PT0S}};
\end{tikzpicture}}
\newcommand{\dnInterval}{\dnfit\begin{tikzpicture}[x=\dnu,y=1cm,baseline=-2.5pt]
  \draw[mx,->](0,0)--(3.7,0);
  \foreach\x in{0.25,0.75,1.25,1.75,2.25,2.75,3.25}{\draw[co](\x,0)--(\x,0.17);}
  \node[mol,anchor=south west]at(0.1,0.2){$=$\,\textsf{P30D}};
  \foreach\x in{0.45,1.2,1.95,2.7,3.45}{\draw[cr](\x,-0.17)--(\x,0);}
  \node[mrl,anchor=north west]at(0.1,-0.2){$=$\,\textsf{P45D}};
\end{tikzpicture}}
\begin{table}[!t]
\centering
\setlength{\tabcolsep}{6pt}
\renewcommand{\arraystretch}{1.2}
\caption{Per-operand denotations and verdicts for the BSB--BnF example (\cref{ex:bsb}). For
each operand, the offer ($\offid{c_i}$, BSB) is drawn above the axis in
\textcolor{cOfferText}{blue} and the request ($\reqid{c'_i}$, BnF) below in
\textcolor{cRequestText}{amber}; filled endpoints are included, hollow ones excluded, an arrow
marks an unbounded ray, the shaded band is the overlap, and the \odrl{timeInterval} row shows
the recurrence combs. Verdict cells are red, green, or grey for \Conflict, \Compatible,
and~\Unknown.}
\label{tab:example-constraints}
{\arrayrulecolor{black!40}%
\begin{tabularx}{\linewidth}{@{} l l C c @{}}
\toprule
\textbf{Constraints} & \textbf{Temporal Operand} & \textbf{Denotation} & \textbf{Verdict}\\
\midrule
$\offid{c_1},\reqid{c'_1}$ & \odrl{dateTime} ($\Inst$) & \dnDate & \vConf\\
\lrule
$\offid{c_2},\reqid{c'_2}$ & \odrl{elapsedTime} ($\Dur$) & \dnElapsed & \vComp\\
\lrule
$\offid{c_3}$ & \odrl{meteredTime} ($\Dur$) & \dnMetered & \vUnk\\
\lrule
$\offid{c_4},\reqid{c'_4}$ & \odrl{delayPeriod} ($\Dur$) & \dnDelay & \vConf\\
\lrule
$\offid{c_5},\reqid{c'_5}$ & \odrl{timeInterval} ($\Dur$) & \dnInterval & \vConf\\
\bottomrule
\end{tabularx}}
\end{table}

%% file: pipe-line.tex
\usetikzlibrary{arrows.meta}
\usetikzlibrary{arrows.meta,decorations.pathreplacing}
\usetikzlibrary{positioning}
\usetikzlibrary{fit,backgrounds}          
\colorlet{cNegBg}{cyan!8}                 
\colorlet{cRtBg}{violet!8}                
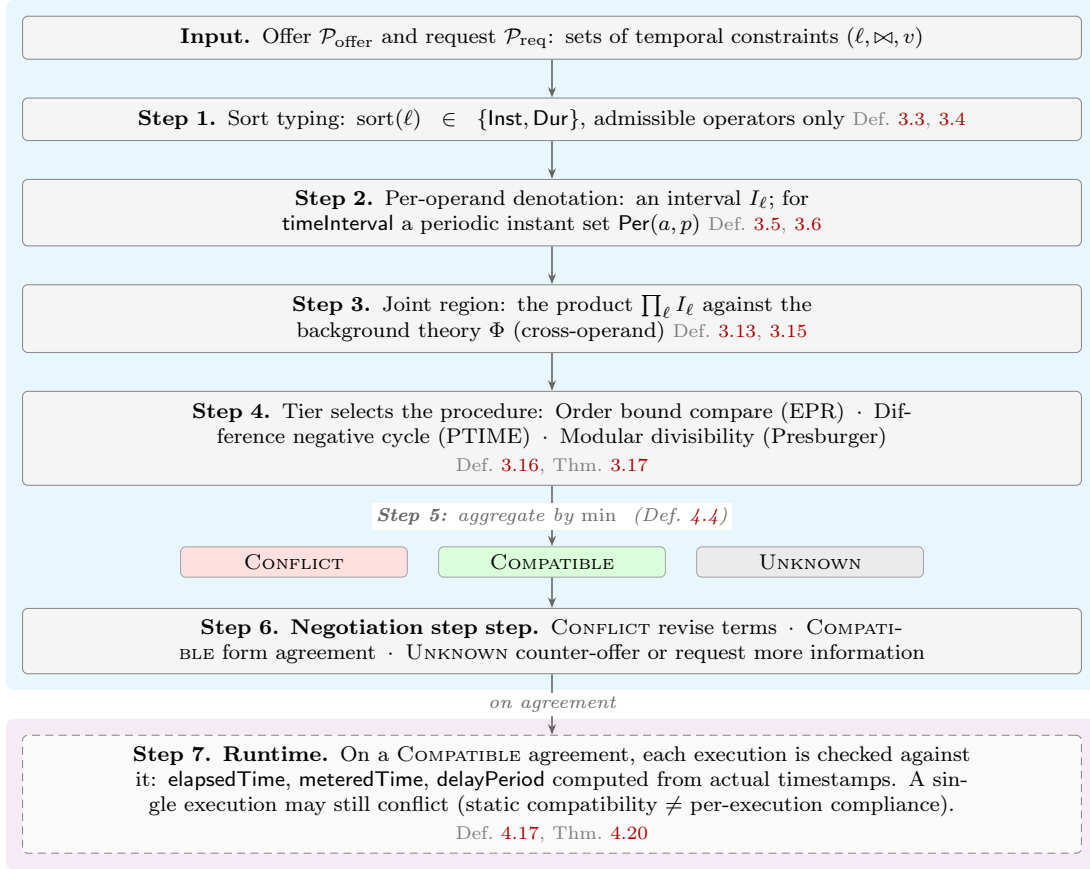
\begin{figure}[t]
\centering
\footnotesize
\begin{tikzpicture}[
  >={Stealth[length=4pt]},
  node distance=5mm,
  stage/.style={draw=black!45,rounded corners=2pt,fill=black!4,
                inner xsep=6pt,inner ysep=4pt,align=center,
                text width=\dimexpr\linewidth-26pt\relax},
  chip/.style={draw=black!35,rounded corners=2pt,inner sep=3pt,
               minimum width=30mm,align=center},
  flow/.style={->,semithick,black!60},
  opl/.style={font=\scriptsize\itshape,color=black!60,fill=white,inner sep=1.5pt},
]
\node[stage] (in)
  {\textbf{Input.} Offer $\PolicyOffer$ and request $\PolicyReq$: sets of temporal constraints $(\ell,\bowtie,v)$};
\node[stage,below=of in] (sort)
  {\textbf{Step 1.} Sort typing: $\sortof(\ell)\in\{\Inst,\Dur\}$, admissible operators only \reff{Def.~\labelcref{def:sort},~\labelcref{def:admissible}}};
\node[stage,below=of sort] (den)
  {\textbf{Step 2.} Per-operand denotation: an interval $I_\ell$; for \odrl{timeInterval} a periodic instant set $\Per(a,p)$ \reff{Def.~\labelcref{def:denotation},~\labelcref{def:ti-rec}}};
\node[stage,below=of den] (reg)
  {\textbf{Step 3.} Joint region: the product $\textstyle\prod_\ell I_\ell$ against the background theory $\Phi$ (cross-operand) \reff{Def.~\labelcref{def:product-denotation},~\labelcref{def:frame}}};
\node[stage,below=of reg] (tier)
  {\textbf{Step 4.} Tier selects the procedure: Order bound compare (EPR)\, $\cdot$\, Difference negative cycle (PTIME)\, $\cdot$\, Modular divisibility (Presburger)\\[1pt]
   \reff{Def.~\labelcref{def:tier}, Thm.~\labelcref{thm:tiered}}};
\node[chip,fill=green!14,below=8mm of tier] (vp) {\Compatible};
\node[chip,fill=red!12,left=4mm of vp]  (vc) {\Conflict};
\node[chip,fill=black!8,right=4mm of vp] (vu) {\Unknown};
\node[stage,below=4mm of vp] (neg)
  {\textbf{\textbf{Step 6.} Negotiation step step.} \Conflict\ revise terms\, $\cdot$\, \Compatible\ form agreement\, $\cdot$\, \Unknown\ counter-offer or request more information};
\node[stage,densely dashed,fill=black!2,below=8mm of neg] (rt)
  {\textbf{\textbf{Step 7.} Runtime.} On a \Compatible\ agreement, each execution is
   checked against it: \odrl{elapsedTime}, \odrl{meteredTime}, \odrl{delayPeriod}
   computed from actual timestamps. A single execution may still conflict
   (static compatibility $\neq$ per-execution compliance).\\[1pt]
   \reff{Def.~\labelcref{def:runtime-conflict}, Thm.~\labelcref{thm:static-runtime}}};
\draw[flow] (in)   -- (sort);
\draw[flow] (sort) -- (den);
\draw[flow] (den)  -- (reg);
\draw[flow] (reg)  -- (tier);
\draw[flow] (tier.south) -- (vp.north)
  node[opl,midway]{\textbf{Step 5:} aggregate by $\min$\ \ ({\scriptsize Def.~\labelcref{def:product-verdict}})};
\draw[flow] (vp.south) -- (neg.north);
\draw[flow] (neg.south) -- (rt.north) node[opl,midway]{on agreement};
\begin{scope}[on background layer]
  \node[fill=cNegBg,rounded corners=3pt,inner sep=6pt,
        fit=(in)(sort)(den)(reg)(tier)(vc)(vp)(vu)(neg)] {};
  \node[fill=cRtBg,rounded corners=3pt,inner sep=6pt,fit=(rt)] {};
\end{scope}
\end{tikzpicture}
\caption{Pipeline of the temporal conflict-detection framework, from a policy pair to a
three-valued verdict and the negotiation step it triggers, and on a compatible agreement
to runtime; the two shaded bands separate the negotiation-time (static) and
runtime phases.}
\label{fig:pipeline}
\end{figure}